\documentclass[letterpaper, 10 pt, conference]{ieeeconf}  

\IEEEoverridecommandlockouts                              

\RequirePackage{etex}
\usepackage{hyperref}
\usepackage{verbatim} 
\usepackage{amsmath} 
\usepackage{amssymb}  
\usepackage{ntheorem}
\usepackage{algorithm}
\usepackage{balance}
\usepackage{mathtools}
\usepackage{siunitx}

\usepackage[usenames,dvipsnames]{xcolor}
\usepackage{psfrag,amsbsy,graphics,float}
\usepackage[dvips]{graphicx}
\usepackage{pgfplots}
\usepackage{pgfplotstable}
\usepgfplotslibrary{fillbetween}
\usepackage{booktabs}
\pgfplotsset{compat = newest}
\usetikzlibrary{arrows,positioning,shapes,intersections,patterns,calc}

\theoremstyle{break}
\newtheorem{mydef}{Definition}
\newtheorem{myremark}{Remark}
\newtheorem{mytheo}{Theorem}
\newtheorem{mylemma}{Lemma}

\newtheorem{myprope}{Property}
\newtheorem{myassum}{Assumption}

\newcommand\tran{\mkern-2mu\raise1.25ex\hbox{$\scriptscriptstyle\top\hspace{0.5mm}$}\mkern-3.5mu}
\newcommand{\R}{\mathbb{R}}
\newcommand{\N}{\mathbb{N}}

\newcommand{\X}{\mathcal{X}}

\newcommand{\bm}[1]{{\boldsymbol{#1}}}
\newcommand{\Verts}[1]{{\Vert #1 \Vert}}
\DeclareMathOperator{\diag}{diag}
\DeclareMathOperator{\var}{var}
\DeclareMathOperator{\mean}{\mu}
\DeclareMathOperator{\Var}{\Sigma}
\DeclareMathOperator{\Mean}{\bm\mu}

\newcommand{\GP}{\mathcal{GP}}

\newcommand{\g}{\bm g}
\newcommand{\ddqd}{\ddot{\bm q}_\bm{d}}
\newcommand{\dqd}{\dot{\bm q}_\bm{d}}
\newcommand{\qd}{\bm{q}_\bm{d}}

\newcommand{\dqe}{\dot{{\bm e}}}
\newcommand{\qe}{{\bm e}}
\newcommand{\ddq}{\ddot{\bm q}}
\newcommand{\dq}{\dot{\bm q}}
\newcommand{\q}{\bm q}

\newcommand{\x}{\bm x}
\newcommand{\y}{\bm y}

\usepackage[noabbrev]{cleveref}
\crefname{myprope}{Property}{Properties}
\crefname{myremark}{Remark}{Remarks}
\crefname{myassum}{Assumption}{Assumptions}
\crefname{myproposition}{Proposition}{Propositions}
\crefname{mycorollary}{Corollary}{Corollaries}
\crefname{mylemma}{Lemma}{Lemmas}
\crefname{mytheo}{Theorem}{Theorems}
\crefname{mydef}{Definition}{Definitions}
\crefname{figure}{Fig.}{Fig.}
\Crefname{figure}{Figure}{Figures}
\crefname{equation}{}{}

\usepackage{autonum}

\title{\LARGE \bf
Stable Gaussian Process based Tracking Control of Lagrangian Systems 
}

\author{Thomas Beckers$^{1}$, Jonas Umlauft$^{1}$, Dana Kuli\'c$^{2}$ and Sandra Hirche$^{1}$
\thanks{$^{1}$ are with the Chair of Information-oriented Control (ITR), Department of Electrical and Computer Engineering,
Technical University of Munich, 80333 Munich, Germany\newline
{\tt\small \{t.beckers, jonas.umlauft, hirche\}@tum.de}\newline
$^{2}$ is with the University of Waterloo, Waterloo, ON N2L 3G1, Canada\newline
{\tt\small dana.kulic@uwaterloo.ca}}
}

\begin{document}

\maketitle
\thispagestyle{empty}
\pagestyle{empty}

\begin{abstract}
High performance tracking control can only be achieved if a good model of the dynamics is available. However, such a model is often difficult to obtain from first order physics only. In this paper, we develop a data-driven control law that ensures closed loop stability of Lagrangian systems.
For this purpose, we use Gaussian Process regression for the feed-forward compensation of the unknown dynamics of the system. The gains of the feedback part are adapted based on the uncertainty of the learned model. Thus, the feedback gains are kept low as long as the learned model describes the true system sufficiently precisely. We show how to select a suitable gain adaption law that incorporates the uncertainty of the model to guarantee a globally bounded tracking error. A simulation with a robot manipulator demonstrates the efficacy of the proposed control law.
\end{abstract}

\section{Introduction}
Lagrangian systems represent a very important and large class of dynamical systems, for which the equations of motion can be derived via the Euler-Lagrange equation. In the past decades, various control schemes for this class have been proposed, most of them can be considered as a subset of computed torque control laws. With computed torque control, it is possible to derive very effective controllers within robust, adaptive and learning control schemes~\cite{siciliano2010robotics}. The controller is separated into a feed-forward and a feedback part. An exact model of the system is necessary for the feed-forward part to compensate the dynamics to achieve a low gain feedback term. This is beneficial in several ways: it favors disturbance attenuation in the presence of noise, avoids the saturation of the actuators, and allows safe interaction~\cite{nguyen2008learning}. Since the accuracy of the compensation depends on the precision of the model, all generalized external forces must be incorporated as precisely as possible~\cite{spong2008robot}. However, an accurate model of these uncertainties is hard to obtain by classical first principles based techniques. Especially in modern applications of Lagrangian systems such as service robotics, the interaction with unstructured and a priori unknown environments further increases the uncertainty.

A suitable approach to avoid the time-consuming or even infeasible physical modeling is provided by Gaussian Process regression (GPR) which is a promising, data-driven learning approach~\cite{deisenroth2015gaussian}. GPR is a supervised learning technique which combines several advantages: It requires only a minimum of prior knowledge for arbitrary complex functions, generalizes well even for small training data sets and has a precise trade-off between fitting the data and smoothing~\cite{rasmussen2006gaussian}. In addition, GPR provides not only a mean function but also a predicted variance, and therefore a measure of the model fidelity based on the distance to the training data. This is a significant benefit also for the feedback part of the control law since the model fidelity can be used to adapt the feedback gains to keep the gains as low as possible. For this purpose, the gains are kept low in state space regions with high model accuracy and increased otherwise.\\

In~\cite{slotine1987adaptive} an online adaptation law for the control of robotic manipulators is presented so that the tracking error converges to zero. However, this approach is based on an underlying parametric model and is limited to dynamics which are linear in terms of a suitably selected set of parameters. In~\cite{chowdhary2015bayesian} a stable feedback linearization with online learned GPR is proposed, but without adapting the feedback gains. The authors of \cite{nguyen2010using} present a hybrid learning approach which incorporates model knowledge. The classical adaptive control approach with varying feedback gains based on the tracking error and the consequences for stability are analyzed in~\cite{ravichandran2004stability}, but without considering the model fidelity. To the best of the authors' knowledge, no result is yet available for Gaussian Process based tracking control with adaptive feedback gains and stability guarantees.\\

The contribution of this paper is a Gaussian Process based control law for Lagrangian systems, that adapts feedback gains based on the model fidelity. For this purpose, the data-driven GPR learns the difference between an estimated model and the true system dynamics from training data. Afterwards, the control law uses the mean of the GPR to compensate the unknown dynamics and the model fidelity to adapt the feedback gains. The derived method guarantees that the tracking error is uniformly ultimately bounded and exponentially converges to a ball for a given probability.\\

The remainder of the paper starts with Section~\ref{sec:def} where we introduce Lagrangian systems and GPR. Section~\ref{sec:modeling} and~\ref{sec:ctrl_law} describe the computation of the model error and the bounded tracking error. The method is validated in Section~\ref{sec:sim}.\\

\textbf{Notation:} Vectors and vector-valued functions are denoted with bold characters. Matrices are described with capital letters. The term~$A_{:,i}$ denotes the i-th column of the matrix~$A$. The expression~$\mathcal{N}(\mu,\Sigma)$ describes a normal distribution with mean~$\mu$ and covariance~$\Sigma$. The Euclidean norm is given by~$\Vert\cdot\Vert$ and the norm of a matrix~$A$~by $\Verts{A}=\bar\lambda(A\tran A)^{1/2}$.
\section{Definitions}\label{sec:def}
This section starts with the necessary background on Gaussian Process regression. Afterwards, the class of Lagrangian system and stability concepts used are introduced.
\subsection{Gaussian Process regression}
Let~$(\Omega, \mathcal{F},P)$ be a probability space with the sample space~$\Omega$, the corresponding~$\sigma$-algebra~$\mathcal{F}$ and the probability measure~$P$. The set~$\X \subseteq \R^n$ with~$n\in\N_{>0}$ denotes the index set. A stochastic process is a discrete or real valued function~$f(\bm x, \omega)$ which is a measurable function of~$\omega\in\Omega$ with~$\bm x\in\X$. A Gaussian Process is a stochastic process with
\begin{align}
f(&\bm x) \sim \GP(m(\bm x),k(\bm x,\bm x^\prime)),\qquad \bm x,\bm x^\prime\in\X\\
m(&\bm x)\colon\X\to\R,\, k(\bm x,\bm x^\prime)\colon\X\times \X\to\R,
\end{align}
which is fully described by a mean function~$m(\bm x)$ and a covariance function~$k(\bm x,\bm x^\prime)$, since with fixed~$\bm x$ it is Gaussian distributed. The mean function is usually defined to be zero if no prior knowledge regarding $f$ is available~\cite{rasmussen2006gaussian}. The covariance function is a measure of the correlation of data two points~$(\bm x,\bm x^\prime)$. The covariance function depends on hyperparameters, which are function dependent. In this paper, we use the Gaussian Process 
\begin{align}
\bm f(\x)&\sim \GP(\bm m(\x),\bm k(\x,\bm x^\prime)),\,\bm f\colon\R^n\to\R^n
\end{align}
with $\X=\R^n$ for regression of vector-valued, nonlinear functions. Since the output of a Gaussian Process is one dimensional, a~$n$-dimensional function~$\bm f$ requires~$n$ GPs. Therefore, the vector valued function~$\bm m(\cdot)=[m_1(\cdot),\ldots,m_n(\cdot)]\tran$ describes the mean functions for each component of~$\bm f(\x)$. The Gaussian Process for each state depends on the corresponding mean and covariance function and is given by 
\begin{align}
\bm f(\x)&=
\begin{cases} 
f_1(\x)\sim \GP(m_1(\x),k_{\varphi_1}(\x,\bm x^\prime))\\
\vdots\hspace{0.9cm}\vdots\hspace{0.5cm}\vdots\\
f_n(\x)\sim \GP(m_n(\x),k_{\varphi_n}(\x,\bm x^\prime)),
\end{cases}
\end{align}
with the set of hyperparameters~$\varphi_i$. The GPR is trained with an input and a corresponding output set which is generated by
\begin{align}
\y&=\bm f(\x)+\bm\eta,\,\bm y\in\R^n\\
\bm\eta&\sim\mathcal{N}(\bm 0,\diag (\sigma_{n,1}^2,\ldots,\sigma_{n,n}^2))
\end{align}
with Gaussian noise~$\bm\eta\in\R^n$. The elements~$\sigma^2_{n,i}$ are the variances of the noise of the output data for all~$i\in\{1,\ldots,n\}$. The~$m$ training inputs~$\{\bm x^{\{j\}}\}_{j=1}^m$ and corresponding outputs~$\{\bm y^{\{j\}}\}_{j=1}^m$ are concatenated in an input training matrix~$X=[\bm x^{\{1\}},\bm x^{\{2\}},\ldots,\bm x^{\{m\}}]\in\R^{n\times m}$ and an output training matrix~$Y=[\bm y^{\{1\}},\bm y^{\{2\}},\ldots,\bm y^{\{m\}}]\tran\in\R^{m\times n}$. Accordingly, the training data set is described by~$\mathcal D=\{X,Y\}$. The prediction at a test input $\x^*\in\R^n$ for each component~$i$ of a new output vector~$y^*_i\in\R$ is calculated as a Gaussian distributed variable with the conditional mean~$\mean(y^*_i\vert \x^*,\mathcal D)$ and the conditional variance~$\var(y^*_i\vert \x^*,\mathcal D)$. The joint distribution of the~$i$-th component of the vector~$y^*_i$ and the corresponding vector of the training outputs~$Y$ is 
\begin{align}
\begin{bmatrix} Y_{:,i} \\ y^*_i \end{bmatrix}\sim \mathcal{N} \left(\bm 0, \begin{bmatrix} K_{\varphi_i}(X,X) & \bm k_{\varphi_i}(\x,X)\\ \bm k_{\varphi_i}(\x,X)\tran & k_{\varphi_i}(\x,\x) \end{bmatrix}\right)\label{for:joint}
\end{align} 
where~$Y_{:,i}$ is the~$i$-th column of the matrix~$Y$. The function~$K_{\varphi_i}(X,X)$ is called the covariance matrix and~$\bm k_{\varphi_i}(\x,X)$ the vector-valued extended covariance function with the set of hyperparameters~$\varphi_i$
\begin{align}
\begin{split}
&K_{\varphi_i}(X,X)\colon\R^m\times \R^m\to\R^{m\times m}\\
&K_{j',j}= k_{\varphi_i}(X_{:, j'},X_{:, j})\\
&\bm k_{\varphi_i}(\x,X)\colon\R\times \R^m\to\R^m,\,k_{\varphi_i,j} = k_{\varphi_i}(\bm{x},X_{:, j})\\
&\forall j',j\in\lbrace 1,\ldots,m\rbrace,i\in\lbrace 1,\ldots,n\rbrace .
\end{split}
\end{align}
Assuming the mean functions of the GPs are set to zero, a prediction of the~$i$-th component of~$\y^*\in\R^n$ is derived from the joint distribution~\eqref{for:joint}. A nonzero mean function can be easily included, see~\cite{rasmussen2006gaussian} for details. The conditional probability distribution is Gaussian with the conditional mean
\begin{align}
\mean_i(\y^*\vert \x^*,\mathcal D)&=\bm k_{\varphi_i}(\x^*,X)\tran (K_{\varphi_i}+I \sigma^2_{n,i})^{-1}Y_{:,i}
\end{align}
where $I$ is the identity matrix while the variance of the prediction is given by
\begin{align}
\var_i(\y^*\vert \x^*,\mathcal D)&=k_{\varphi_i}(\x^*,\x^*)-\bm k_{\varphi_i}(\x^*,X)\tran \notag\\
& \phantom{{}=}(K_{\varphi_i}+I \sigma^2_{n,i})^{-1} \bm k_{\varphi_i}(\x^*,X).
\end{align}
The set of hyperparameters~$\varphi_i$ are optimized by means of the likelihood function, thus by 
\begin{align}
\varphi_i^* = \arg\max_{\varphi_i} \log P(Y_{:,i}|X,\varphi_i),\,\forall i=1,\ldots,n.
\end{align}
The~$n$ normally distributed components of~$y^*_i\vert \x^*,\mathcal D$ are combined in a multi-variable Gaussian distribution 
\begin{align}
\y^*\vert \x^*,\mathcal D &\sim \mathcal{N} (\bm\mean(\cdot),\Var(\cdot))\\
\bm \mean(\y^*\vert \x^*,\mathcal D)&=[\mean_1(\cdot),\ldots,\mean_n(\cdot)]\tran\\
\Var(\y^*\vert \x^*,\mathcal D)&=\diag(\var_1(\cdot),\ldots,\var_n(\cdot)).
\end{align} 
The predicted variance with respect to just a part of~$\x^*\in\R^n$ can be done by marginalization. Assume~$\x^*=[\x^{*}_1\tran,\x^{*}_2\tran]\tran$ with~$\x^*_1\in\R^{n_1},\x^*_2\in\R^{n_2}$ and~$\x^*\in\R^{n=n_1+n_2}$. Then, the marginal variance of the prediction based on~$\x^*_1$ is given by
\begin{align}
\var_i(\y^*\vert \x^*_1,\mathcal D)&=k_{\phi_i}(\x^*_1,\x^*_1)-\bm k_{\phi_i}(\x^*_1,X_{1:n_1,:})\tran \notag\\
& \phantom{{}=}(K_{\varphi_i}+I \sigma^2_{n,i})^{-1} \bm k_{\phi_i}(\x^*_1,X_{1:n_1,:})\label{for:margvar}
\end{align}
with the covariance function defined on the input space~$\R^{n_1}$. Since the dimension of the input data is reduced, i.e. from~$n_1+n_2$ to $n_1$, the necessary set of hyperparameters~$\phi_{i}$ is a subset of the original set of hyperparameters~$\varphi_{i}$. The combined marginal variance is rewritten as
\begin{align}
\Var(\y^*\vert \x^*_1,\mathcal D)&=\diag(\var_1(\cdot),\ldots,\var_{n_1}(\cdot)).
\end{align}
\subsection{Class of Lagrangian systems}
The assumed class of Lagrangian systems is described by 
\begin{align}
\label{for:dyn_model_man}
H(\q)\ddq+C(\q,\dq)\dq+\bm \g(\q)+\bm\kappa(\breve{\q})=\bm\tau,
\end{align}
where~$\q\in\R^n$ are the generalized coordinates with their time derivatives~$\ddq,\dq\in\R^n$ and the generalized input~$\bm\tau\in\R^n$. An additional unknown dynamic~$\bm\kappa(\breve{\q}):\R^{3n}\to\R^{n}$ which\linebreak depends on $\breve{\q}=[\ddq^\top,\dq^\top,\q^\top]^\top$ affects the system as a generalized force. The generalized inertia~$H(\q)\colon\R^n\to\R^{n\times n}$, the matrix $C(\q,\dq)\colon\R^n\times \R^n\to\R^{n\times n}$, and~$\g(\q)\colon\R^n\to\R^n$ are assumed to have the following properties:
\begin{myprope}[Structural properties]\label{prop:structural}
The matrix~$H(\q)$ is symmetric and positive definite and there is a parameterization of matrix~$C(\q,\dq)$ so that $\forall \dq,\q\in\R^n$
\begin{itemize}
\item~$\dot{H}(\q)=C(\q,\dq)+C(\q,\dq)\tran\in\R^{n\times n}$ and thus,
\item~$\dot{H}(\q)-2C(\q,\dq)\in\R^{n\times n}$ is skew-symmetric.
\end{itemize}
\end{myprope}
\begin{myprope}[Boundedness and Linearity]\label{prop:bound}
\begin{itemize}
\item The matrix~$H(\q)$ is bounded, i.e. there exists two constants~$h_1\in\R_{>0}$ and $h_2\in\R_{>0}$ such that $H(\q)$ is bounded with~$h_1\Verts{\x}^2\leq\x\tran H(\q)\x\leq h_2\Verts{\x}^2$ for all~$\x,\q \in \R^n$.
\item The matrix~$C(\q,\dq)$ is bounded in~$\q$ and linear in~$\dq$,i .e. there exists a$~k_C\in\R_{>0}$ such that~$\Vert C(\q,\dq) \Vert\leq k_C\Vert \dq \Vert$, and~$C(\q,\dq)\bm p=C(\q,\bm p)\dq$ for all~$\q,\dq,\bm p \in\R^n$.
\end{itemize}
\end{myprope}
As shown in \cite{ghorbel1993positive}, these properties hold, for example, for the class of robotic manipulators with revolute joints.
\section{Modeling}
\label{sec:modeling}
In this section, we introduce the modeling procedure and the error estimation between the learned and the true dynamics. Consider the Lagrangian system in~\cref{for:dyn_model_man} for which the following is assumed:
\begin{myassum}\label{as:rkhs}
The generalized external force~$\bm\kappa(\breve{\q})$ has a bounded reproducing kernel Hilbert space (RKHS) norm, i.e.~$\Verts{\bm\kappa}_k<\infty$, with respect to the covariance function~$k(x,x^\prime)$ of a GPR.
\end{myassum}
This assumption ensures that the covariance function is suitable for the approximation of the function~$\bm\kappa(\breve{\q})$. A comparison of different covariance functions and recommendation for Lagrangian systems can be found in~\cite{bishop2006pattern,cheng2016learning}.\\
To achieve a controller with a good feed-forward compensation, the system~\cref{for:dyn_model_man} must be identified. Since partial a priori knowledge of the system is often available, we use a so-called hybrid learning approach which is a combination of a parametric and a data-driven model. For this purpose, a Gaussian Process learns the difference between the real and the estimated dynamics 
\begin{align}
\hat{\bm \tau}&=\hat H(\q)\ddq+\hat C(\q,\dq)\dq+\hat \g(\q)\label{for:learn1}
\end{align}
of the system where $\hat H, \hat C$ and $\hat\g$ are the estimates of the true matrices, thus
\begin{align}
\tilde{\bm \tau} =\bm \tau-\hat{\bm \tau} = &H(\q)\ddq+C(\q,\dq)\dq+\g(\q)+\bm\kappa(\breve{\q})\\
-&\hat H(\q)\ddq-\hat C(\q,\dq)\dq-\hat \g(\q).\label{for:learn}
\end{align} 
Without loss of generality, we assume in the following analysis that~$H=\hat H,C= \hat C,$ and~$g=\hat g$ since the error of the estimation can be included in~$\bm\kappa(\breve{\q})$.
For the generation of training data, a set of~$\{\ddq,\dq,\q\}$ as training inputs~$X$ and~$\tilde{\bm \tau}$ as training output~$Y$ is necessary, which could be generated through any suitably well-behaving control. For the stability analysis of the closed loop, the error between the model and the true dynamics is necessary. A probabilistic upper bound for the distance between the mean prediction~$\Mean(\tilde{\bm \tau})$ of the Gaussian Process regression and the true function is given in \cite{srinivas2012information} and is extended for multidimensional functions in the following lemma.
\begin{mylemma}
\label{lemma:boundederror}
Consider a Lagrangian System~\cref{for:dyn_model_man} and the estimated dynamics~\cref{for:learn1}. A Gaussian Process is trained with the difference between the true and the estimated dynamics. Then the model error is bounded by
\begin{align}
P\{\Verts{\Mean(\tilde{\bm \tau}\vert \breve{\q},\mathcal D)-\bm\kappa(\breve{\q})}\leq \Verts{\bm \beta\tran \Var^{\frac{1}{2}}(\tilde{\bm \tau}\vert \breve{\q},\mathcal D)}\}\geq (1-\delta)^n
\end{align}
for~$\breve{\q}\in D$ on a compact set $D\subset\R^{3n}$ where each element of~$\bm\beta\in\R^{n}$ is defined as
\begin{align}
\beta_j&=\sqrt{2\Verts{\kappa_j}^2_k+300 \gamma_j \ln^3\left(\frac{m+1}{\delta}\right)}.
\end{align}
The variable~$\gamma_j\in\R$ is the maximum information gain 
\begin{align}
\gamma_j&=\max_{\breve{\q}^{\{1\}},\ldots,\breve{\q}^{\{m+1\}}\in D}\frac{1}{2}\log \vert I+\sigma^{-2}K_{\varphi_j}(\x,\x^\prime)\vert\\
&\x,\x^\prime\in\left\lbrace\breve{\q}^{\{1\}},\ldots,\breve{\q}^{\{m+1\}}\right\rbrace
\end{align}
with the covariance matrix~$K_{\varphi_j}$.
\end{mylemma}
\begin{proof}
The result is a direct consequence of \cite[Theorem 6]{srinivas2012information} and the fact that~$\bm\eta$ is uncorrelated. Therefore, 
\begin{align}
\text{P}&\left\lbrace \forall \breve{\q}\in D,\right.\notag\\
&\vert\mean(\tilde{\tau}_1\vert \breve{\q},\mathcal D)-f_1(\breve{\q})\vert\leq \vert \beta_1 \var_1^{\frac{1}{2}}(\tilde{\tau}_1\vert \breve{\q},\mathcal D)\cap\ldots\cap\notag\\
&\vert\mean(\tilde{\tau}_n\vert \breve{\q},\mathcal D)-f_n(\breve{\q})\vert\leq \vert \beta_n \var_n^{\frac{1}{2}}(\tilde{\tau}_n\vert \breve{\q},\mathcal D)\vert\left.\right\rbrace\geq (1-\delta)^n\notag\\
\Rightarrow&\text{P}\left\lbrace \forall \breve{\q}\in D,\, \Delta\leq \Verts{\bm \beta\tran \Var^{\frac{1}{2}}(\tilde{\bm \tau}\vert \breve{\q},\mathcal D)}\right\rbrace\geq (1-\delta)^n\label{for:prob}
\end{align}
provides an upper bound for the norm of the model error with a probability of at least~$(1-\delta)^n$.
\end{proof}
The information capacity~$\gamma$ has a sublinear dependence on the number of training points for many commonly used covariance functions and can be approximated by a constant, e.g. shown in \cite{srinivas2012information}. Therefore, even though~$\bm \beta$ is increasing with the number of training samples, it is possible to learn the true function~$\bm f(\x)$ arbitrarily exactly.\\

\section{PD control with variable Gains}\label{sec:ctrl_law}
Classical computed torque control employs static feedback gains. Low feedback gains are desirable to avoid saturation of the actuators and achieve a good noise attenuation. However, the unknown dynamics usually requires a specific minimal feedback gain to keep the tracking error under a defined limit. After a training procedure, we use the mean of the Gaussian Process regression to compensate~$\bm\kappa(\breve{\q})$ with the feed-forward part and adapt the gains based on the model fidelity. For this purpose, the uncertainty of the regression is used to scale the feedback gains.\\
We start with the following natural assumption for the desired trajectory.
\begin{myassum}
\label{as:tra}
The desired trajectory is bounded by~$\Vert\qd\Vert\leq \bar{q}_d$,~$\Vert\dqd\Vert\leq \bar{\dot q}_d$ with~$\bar{q}_d,\bar{\dot q}_d\in\R_{\geq 0}$.
\end{myassum}
In the next step, an assumption for the varying gains is introduced.
\begin{myassum}
\label{as:KdKp}
Let~$G_d(\dq,\q)\colon\R^{n}\times\R^n \rightarrow\R^{n\times n}$ and~$G_p(\q)\colon\R^{n} \rightarrow\R^{n\times n}$.\\
i) Let~$K_d\colon\R^{n\times n} \rightarrow\R^{n\times n}$ be a positive definite and symmetric matrix such that $(K_d\circ G_d)$ is continuous and that there exits a lower and upper quadratic bound
\begin{align}
k_{d1}\Vert \x \Vert^2\leq \x\tran K_d(G_d(\dq,\q))\bm x\leq k_{d2}\Vert \x \Vert^2,\,\forall \dq,\q,\x\in\R^n,
\end{align}
with~$k_{d1}\in\R_{\geq 0}$ and~$k_{d2}\in\R_{>0}$.\\
ii) Let~$K_p\colon\R^{n\times n} \rightarrow\R^{n\times n}$ be a positive definite diagonal matrix. Each diagonal element~$K_{p,ii}(G_p(\q))$ is continuous and bounded by~$0<\underline{k}_{p,ii}\leq K_{p,ii}(G_p(\q))) \leq \bar{k}_{p,ii}$ for all~$\q\in\R,i=1,\ldots,n$, so that
\begin{align}
k_{p1}\Vert \x \Vert^2\leq \x\tran K_p(G_p(\q))\bm x\leq k_{p2}\Vert \x \Vert^2,\,\forall \q,\x\in\R^n
\end{align}
with~$k_{p1}=\min_{i}\underline{k}_{p,ii}\in\R_{\geq 0}$ and~$k_{p2}=\max_{i}\bar{k}_{p,ii}\in\R_{>0}$ for all~$i=1,\ldots,n$.
\end{myassum}
~\Cref{as:KdKp} restricts the matrix $K_p$ to be diagonal which results in a decentralized feedback of the tracking error. The symmetric form of $K_p$ and $K_d$ is a common assumption which does not restrict the applicability of the approach but must be kept in mind during the design of the controller. Before the control law is proposed, the following definition and lemma are introduced.
\begin{mydef}
Assume a Gaussian Process trained with the difference between the true and the estimated dynamics of a Lagrangian system~\cref{for:learn}. The marginal variances~$\Var_d:\R^n\times\R^n\to\R^{n\times n}$ and~$\Var_p:\R^n\to\R^{n\times n}$ are defined with~\cref{for:margvar} by
\begin{align}
\Var_d=\Var(\tilde{\bm \tau}\vert \dq,\q,\mathcal D),\,\Var_p&=\Var(\tilde{\bm \tau}\vert\q,\mathcal D).
\end{align}
\end{mydef}
\begin{mylemma}
\label{lemma:A}
There exists an~$\epsilon>0$ such that the matrix~$A\in\R^{2n\times 2n}$ given by\footnote{For notational convenience, the dependencies of $H,C,\bm g$ are dropped here.}
\begin{align}
A=\begin{bmatrix}
-K_d(\Var_d)+\varepsilon H & \frac{\varepsilon}{2}(- K_d\tran(\Var_d)+ C)\\
\frac{\varepsilon}{2}(- K_d(\Var_d)+ C\tran) & - \varepsilon K_p(\Var_p)
\end{bmatrix},
\end{align}
which is is negative definite under~\cref{prop:structural} and~\cref{as:KdKp} for all~$\dq,\q\in\R^n$.
\end{mylemma}
\begin{proof}
According to the Schur's Lemma,~$A$ is negative definite if
\begin{align}
A_{11}&=-K_d(\Var_d)+\varepsilon H\\
S&=-\varepsilon K_p(\Var_p)+\frac{\varepsilon^2}{4}(K_d(\Var_d)-C\tran)\notag\\
&\phantom{=}(K_d(\Var_d)-\varepsilon H)^{-1}(K_d\tran(\Var_d)- C)
\end{align}
are negative definite where~$A_{11}\in\R^{n\times n}$ is the upper left block of~$A$ and~$S\in\R^{n\times n}$ is the Schur complement. Since~$K_d,H$, and~$K_p$ are positive definite and bounded,~$\varepsilon$ can be chosen sufficiently small to obtain the negative definiteness of~$A_{11}$. The second summand of the Schur complement~$S$ is quadratic in~$\varepsilon$ and positive definite while the first summand is linear in~$\varepsilon$ and negative. Thus, for every~$\q,\dq\in\R^n$, an~$\varepsilon$ can be found which guarantees the negative definiteness of the Schur complement. Therefore, the matrix~$A$ is negative definite. 
\end{proof}
The next theorem introduces the control law with guaranteed boundedness of the tracking error which is the main contribution of the paper.
\begin{mytheo}
\label{theo:main}
Consider the Lagrangian system~\cref{for:dyn_model_man} which satisfies the~\cref{prop:structural,prop:bound} and~\cref{as:rkhs,as:KdKp,as:tra}. A Gaussian Process is trained with~$m$ data pairs of the set~$\mathcal D=\{\breve{\q}^{\{i\}},\tilde{\bm\tau}^{\{i\}}\}_{i=1}^{m}$ with
\begin{align}
\tilde{\bm \tau} =\bm \tau-H(\q)\ddq-C(\q,\dq)\dq-\g(\q).
\end{align}
Let~$\qe=\q-\qd,\dqe=\dq-\dqd$ be the tracking error. The control law
\begin{align}
\bm{\tau}&=H(\q)\ddqd+C(\q,\dq)\dqd+\g(\q)+\Mean(\tilde{\bm \tau}\vert \breve{\q},\mathcal D)\notag\\
&-K_d(\Var_d) \dqe-K_p(\Var_p) \qe\label{for:control_law}
\end{align}
guarantees that the tracking error is uniformly ultimately bounded and exponentially converges to a ball with a probability of at least~$(1-\delta)^n,\delta\in(0,1)$.
\end{mytheo}
\begin{proof}
For the stability analysis, we use the following Lyapunov candidate 
\begin{align}
V(\dqe,\qe)=\frac{1}{2}\dqe\tran H(\q)\dqe+\int_0^\qe \bm z\tran K_p(\Var_p)d\bm z+\varepsilon \qe\tran H(\q)\dqe\label{for:lyap}
\end{align}
with a parameter~$\varepsilon>0$. To ensure that the Lyapunov candidate is positive definite, we analyze the domain of the integral
\begin{align}
\int_0^\qe \bm z\tran K_p(\Var_p)d\bm z&\geq \frac{1}{2}\sum_{i=1}^n \underline{k}_{p,ii}e_i^2\geq \frac{1}{2} k_{p1} \Vert \qe \Vert^2,
\end{align}
where each component of the sum has a lower bound and thus, the whole integral is lower bounded. An upper quadratic bound can be found in a similar way as presented in \cite{ravichandran2004stability}. Since the integral is lower bounded and~$H(\q)$ is always positive definite, the parameter~$\varepsilon$ can be chosen sufficiently small to achieve a positive definite and radially unbounded Lyapunov function. The valid interval for~$\varepsilon$ can be determined by the lower bound of the Lyapunov function
\begin{align}
V(\dqe,\qe)\geq \frac{1}{2}h_1\Verts{\dqe}^2+\frac{1}{2}k_{p1}\Verts{\qe}^2-\frac{1}{2}\varepsilon h_2\left( \Verts{\dqe}^2+\Verts{\qe}^2\right)
\end{align}
which is positive for $0<\varepsilon < \min\left\lbrace k_{p1}/h_2,h_1/h_2\right\rbrace$. In the next step, we investigate the time derivative of the Lyapunov function to establish stability of the closed loop. With~\cref{prop:structural} and~\cref{for:control_law}, it can be written as
\begin{align}
\dot V=&\begin{bmatrix} \dqe\\\qe \end{bmatrix}\tran\underbrace{\begin{bmatrix}
-K_d(\Var_d)+\varepsilon H & \frac{\varepsilon}{2}(- K_d\tran(\Var_d)+ C)\\
\frac{\varepsilon}{2}(- K_d(\Var_d)+ C\tran) & - \varepsilon K_p(\Var_p) 
\end{bmatrix}}_{A}\begin{bmatrix} \dqe\\\qe \end{bmatrix}\notag\\
+&\begin{bmatrix} \dqe\tran & \varepsilon\qe\tran \end{bmatrix}\underbrace{\begin{bmatrix}
\Mean(\tilde{\bm \tau})-\bm\kappa(\breve{\q})\\
\Mean(\tilde{\bm \tau})-\bm\kappa(\breve{\q})\label{for:dotV}\\
\end{bmatrix}}_{\bm b}.
\end{align}
For the analysis, we define the two parts of the equation as~$A\in\R^{2n\times 2n}$ and~$\bm b\in\R^{2n}$. The following statements can be made for the matrix~$A$: The submatrix~$A_{11}$ is bounded with
\begin{align}
\dqe\tran A_{11}\dqe&=\dqe\tran\left(-K_d(\Var_d)\varepsilon H\right)\dqe\leq(-k_{d1}+\varepsilon h_2)\Verts{\dqe}^2.
\end{align} 
With~\cref{as:KdKp,as:tra}, and~\cref{prop:bound}, the submatrix~$A_{21}$ is bounded by
\begin{align}
\qe\tran A_{21} \dqe&\leq\varepsilon \left( k_C\Vert\dqe\Vert+k_C \bar{\dot q}_d+ k_{d2}\right) \Vert \dqe \Vert \Vert \qe \Vert.
\end{align}
Then, using~\cref{lemma:boundederror,as:KdKp}, the overall upper bound for the time derivative of the Lyapunov function is given by
\begin{align}
\dot V(\dqe,\qe)&\leq (\varepsilon h_2-k_{d1}) \Vert \dqe \Vert^2-\varepsilon k_{p1}\Vert \qe\Vert^2\notag\\
&+\varepsilon \left( k_C\Vert\dqe\Vert+k_C \bar{\dot q}_d+ k_{d2}\right) \Verts{\dqe}\Verts{\qe}\notag\\
&+(\Verts{\dqe}+\varepsilon\Verts{\qe})\Verts{\bm \beta\tran \Var(\tilde{\bm \tau}\vert \breve{\q},\mathcal D)}. \label{for:Vbound}
\end{align}
Considering the inequality
\begin{align}
\Verts{\dqe}\Verts{\qe}\leq \frac{1}{2}\left( \rho \Verts{\dqe}^2+\frac{\qe^2}{\rho}\right)
\end{align}
that holds for all~$\dqe,\qe\in\R^n$ and~$\rho\in\R_{\geq 0}$,~\cref{for:Vbound} can be rewritten as
\begin{align}
\dot V(\dqe,\qe)&\leq \left(\varepsilon h_2-k_{d1}+\frac{\varepsilon \rho}{2}(k_C \bar{\dot q}_d+k_{d2})\right) \Verts{\dqe}^2\notag\\
&-\varepsilon k_{p1}\frac{\varepsilon_2}{1+\varepsilon_2}\Verts{\qe}^2+\varepsilon k_C\Verts{\dqe}^2\Verts{\qe}\notag\\
&+(\Verts{\dqe}+\varepsilon\Verts{\qe})\Verts{\bm \beta\tran \Var(\tilde{\bm \tau}\vert \breve{\q},\mathcal D)}\label{for:Vbound3}\\
\text{with }&\rho=(1+\varepsilon_2)\frac{k_C \bar{\dot q}_d+ k_{d2}}{2 k_{p1}},\,\varepsilon_2\in\R_{> 0}.\notag
\end{align}
The choice of~$\rho$ guarantees that the factors of the quadratic parts are still negative. The linear part of~\cref{for:Vbound3} can be bounded by a quadratic function with $v_1\Verts{\x}\leq v_1^2/v_2+v_2/4\Verts{\x}^2$ that holds for all~$\x\in\R^n$ and~$v_1,v_2\in\R_{\geq 0}$
\begin{align}
&(\Verts{\dqe}+\varepsilon\Verts{\qe})\Verts{\bm \beta\tran \Var(\tilde{\bm \tau}\vert \breve{\q},\mathcal D)}\notag\\
\leq &\frac{{\bar \Delta}^2 }{v_1}+\frac{v_1}{4}\Verts{\dqe}+\frac{{\varepsilon^2\bar \Delta}^2 }{\varepsilon v_2}+\frac{\varepsilon v_2}{4}\Verts{\qe}.\label{for:mod_bound}
\end{align}
with~$\Verts{\bm \beta\tran \Var(\tilde{\bm \tau}\vert \breve{\q},\mathcal D)}\leq \bar \Delta\in\R_{>0}$ of~\cref{lemma:boundederror} and
\begin{align}
v_1&:=-\varepsilon h_2+k_{d1}-\frac{\varepsilon \rho}{2}(k_C \bar{\dot q}_d+k_{d2}),\,v_2:= k_{p1}\frac{\varepsilon_2}{1+\varepsilon_2}.
\end{align}
Since the covariance function is bounded on the closed interval $D$, the variance~$\Var(\tilde{\bm \tau}\vert \breve{\q},\mathcal D)$ is bounded, for more details see \cite{beckers:cdc2016}. Thus, there exists an upper bound~$\bar \Delta$ for the model error relating to~$\dqe,\qe$.
To ensure that the variables~$v_1,v_2$ are positive, the restriction for $\varepsilon$ must be extended to
\begin{align}
0<\varepsilon <\min\left\lbrace \frac{k_{p1}}{h_2},\frac{h_1}{h_2},\frac{2 k_{d1}}{2 h_2+\rho(k_C \bar{\dot q}_d+k_{d2})}\right\rbrace.\label{for:epsilon2}
\end{align}
With~\cref{for:mod_bound}, equation~\cref{for:Vbound3} can be rewritten as
\begin{align}
\dot V(\dqe,\qe)&\leq -\frac{3}{4}v_1 \Verts{\dqe}^2-\frac{3}{4}\varepsilon v_2 \Verts{\qe}+\varepsilon k_C\Verts{\dqe}^2\Verts{\qe}\notag\\
&+ \frac{{\bar \Delta}^2 }{v_1}+ \varepsilon\frac{{\bar \Delta}^2 }{v_2}.
\end{align}
According to \cite{ravichandran2004stability} and~\cref{lemma:boundederror}, there exists a~$\xi\in\R_{\geq 0}$ and a~$\varrho\in\R_{\geq 0}$ for~\cref{for:Vbound} such that
\begin{align}
\text{P}\left\lbrace\dot V(\dqe,\qe)\leq -\xi V(\dqe,\qe)+\varrho\right\rbrace\geq (1-\delta)^n
\end{align}
holds with
	\begin{align}
		\xi&=\frac{2}{3}\frac{\min\left\lbrace\varepsilon v_2,v_1-\frac{4}{3}\varepsilon k_c \sqrt{\frac{ 2V_0}{k_{p1}-\varepsilon h_2}}\right\rbrace}{\max\left\lbrace\varepsilon h_2+k_{p2},(1+\varepsilon)h_2\right\rbrace}\label{for:xi}\\
		\varrho&=\frac{{\bar \Delta}^2 }{v_1}+ \varepsilon\frac{{\bar \Delta}^2 }{v_2}
	\end{align}
	and the extension of~\cref{for:epsilon2}
	\begin{align}
		0<\varepsilon < \min&\left\lbrace \frac{k_{p1}}{h_2},\frac{h_1}{h_2},\frac{2 k_{d1}}{2 h_2+\frac{2k_{p1}\rho^2}{1+\varepsilon_2}+\frac{8}{3}k_c\sqrt{\frac{ 2V_0}{k_{p1}-\varepsilon h_2}}}\right\rbrace .
	\end{align}
	Finally,~\cref{lemma:boundederror} requires that~$[\ddq,\dq,\q]$ is always an element of the set~$D$. Therefore, it must be chosen so that
\begin{align}
\{\forall\ddq,\dq,\q\in\R^n:V(\dqe,\qe)\leq V_0\}&\in D\\
\text{and }\{\forall\ddq,\dq,\q\in\R^n:\Verts{\dqe^\top,\qe^\top}\leq \sqrt{2\varrho /(\xi\underline{c})}\}&\in D
\end{align}
with $\underline{c}=\min\left\lbrace k_{p1}-\varepsilon h_2, h_1-\varepsilon h_2 \right\rbrace$ which guarantees that the trajectory stays inside the set~$D$,~\cite{ravichandran2004stability}.
Thus, the tracking error is uniformly ultimately bounded and converge to a ball with a probability of at least~$(1-\delta)^n$.
\end{proof}
\begin{myremark}
The first summand of~\cref{for:dotV} contains the influence of the controller on the system while the second summand captures the model error. If a perfect model was available, such that~$\Mean(\tilde{\bm \tau})=\bm\kappa(\breve{\q})$, equation~\cref{for:dotV} with~\cref{lemma:A} would show that the closed loop system is asymptotically stable.
\end{myremark}
\begin{myremark}
A similar idea of GPR-based computed torque control is presented in~\cite{alberto2014computed}, however, without stability analysis.
\end{myremark}
\section{Simulation}\label{sec:sim}
For the simulation, we apply Lagrange's equations to the common model of a 2-link planar manipulator~\cite{murray1994mathematical}. We assume a point mass for the links of~$\SI{1}{\kilogram}$ which are located in the center of each link. The length of the links is set to~$\SI{1}{\meter}$. The joints are without mass and not influenced by any friction. Gravity is assumed to be~$\SI{10}{\meter\second^{-2}}$. 
Here, the generalized coordinates~$q_1$ and~$q_2$ are the joint angles. The initial values are set to $\q_0=[0,0]\tran$. The unknown dynamics~$\bm\kappa(\cdot)$ is simulated by a sample path of a Gaussian Process with a squared exponential covariance function which is acting here as ground truth. This approach guarantees that~$\bm\kappa(\cdot)$ has a bounded RKHS norm regarding to the squared exponential function. This Gaussian Process is trained by 50 data values of the arbitrary chosen nonlinear function
\begin{align}
f(\q,\dq)=
\begin{bmatrix}
-\dot{q}_1+2\sin(q_2)+\vert q_1 \vert\\
-\dot{q}_2+2\sin(q_2)
\end{bmatrix}.\label{for:exdyn}
\end{align}
Now, the proposed control law of~\cref{theo:main} is used. We assume that the estimated matrices of the Lagrangian system are equal to the true matrices. A Gaussian Process with a squared exponential covariance function learns the difference between the estimated model and the true system, thus the unknown dynamics~$\bm\kappa(\cdot)$. For this purpose, we generate 225 pairs of~$\{\tilde\tau\}$ and states~$\{\ddq,\dq,\q\}$ as training data on the domain~$\ddq,\dq\in[-1,1]^2,\,\q\in[0,1]$ to generate a set~$\mathcal D$ of training points. The hyperparameters are optimized by means of the likelihood function. The desired trajectory is a sinusoidal function with $\q_0=[0,1]\tran$.\\
As comparison, the proposed control law with static gains~$K_{p,\text{static}}=\diag(10,10)$ and~$K_{d,\text{static}}=\diag(10,10)$ is used.~\Cref{fig:var_gains} shows the resulting trajectory for the first joint. The system trajectory with static gains (red dotted) is close to the desired trajectory (blue dashed) while it is in the neighborhood of the training data. Outside this area, the tracking error increases. Now, the same  control law of~\cref{theo:main} with variable gains is used. In this example, the gains are adapted according to~\cref{as:KdKp} with
\begin{align}
K_p(\Var_p)&=\diag(10 + 30\var_{p,1}(\tilde{\bm \tau}),10 + 30\var_{p,2}(\tilde{\bm \tau}))\\
K_d(\Var_d)&=\diag(10 + 30\var_{d,1}(\tilde{\bm \tau}),10 + 30\var_{d,2}(\tilde{\bm \tau})).
\end{align}
In~\cref{fig:var_gains}, the color of the trajectory indicates the norm of the current feedback gains. In the area close to the training data, the feedback gains remain low (blue color) while outside the training area the gains increase (red color). The result is that the tracking error is kept low and bounded even for areas where no training data is available. 
\section*{Conclusion}
In this paper, we propose a GPR-based control law for Lagrangian systems which guarantees a bounded tracking error. The feedback gains of the control law are adapted by the model fidelity to keep the feedback gains as low as possible. The main contribution is that the tracking error of the closed loop system with the data-driven GPR model is proven to be uniformly ultimately bounded and exponentially convergent to a ball with a given probability. 

\section*{ACKNOWLEDGMENTS}
The research leading to these results has received funding from the European Research Council under the European Union Seventh Framework Program (FP7/2007-2013) / ERC Starting Grant ``Control based on Human Models (con-humo)'' agreement n\textsuperscript{o}337654.
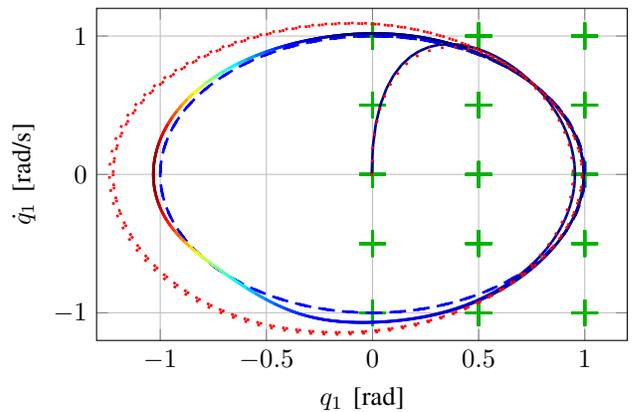
\begin{figure}[t]
\begin{center}
 \vspace{0.2cm}	
 \begin{tikzpicture}
\pgfplotsset{
  set layers,
  mark layer=axis tick labels
}
\begin{axis}[
  xlabel={$q_1$ [rad]},
  ylabel={$\dot{q}_1$ [rad/s]},
  line width=1pt,
  grid = major,
  colormap/jet,
    height=6cm,
  width=\columnwidth,
  xmin=-1.3, xmax=1.2, ymin=-1.2, ymax=1.2]
\addplot[mark=+,color=green!70!black, only marks,mark size=5,line width=1pt] table [x index=0,y index=1]{data/training_points.dat};
\addplot+[color=blue,style=loosely dashed,no marks,line width=1pt] table [x index=0,y index=1]{data/var_gains.dat};
\addplot+[mesh,point meta=\thisrowno{4}, no marks,line width=1pt, shader=interp] table [x index=2,y index=3]{data/var_gains.dat};
\addplot+[color=red,no marks, style=loosely dotted, line width=1pt] table [x index=2,y index=3]{data/const_gains.dat};
\legend{}
\end{axis}
\end{tikzpicture} 
	\vspace{-0.2cm}\caption{Comparsion between the proposed control law with static feedback gains (red dotted) and adapted feedback gains (solid line) regarding to the desired trajectory (blue dashed). The color of the adapted feedback trajectory indicates the norm of the current feedback gains. The green crosses mark the training data points.\label{fig:var_gains}\vspace{-0.5cm}}
\end{center}
\end{figure}

\bibliography{mybib}
\bibliographystyle{ieeetr}

\end{document}